\documentclass{article}
\usepackage{amsthm}
\usepackage{amsmath}
\usepackage{amssymb}
\usepackage{graphicx}
\usepackage{natbib}
\usepackage[utf8]{inputenc}
\usepackage[T1]{fontenc}
\usepackage{url}
\usepackage{hyperref}
\usepackage{tikz}
\usetikzlibrary{automata, positioning, arrows}

\newtheorem{theorem}{Theorem}

\def\bbN{\mathbb{N}}
\def\bbR{\mathbb{R}}

\def\ra{\rightarrow}

\def\bbN{{\mathbb N}}

\def\bbR{{\mathbb R}}

\def\bbZ{{\mathbb Z}}

\def\ra{\rightarrow}

\begin{document}

\title{Coordination Sequences of Periodic Structures are Rational via Automata Theory}
\author{Eryk Kopczyński \\ { \small Institute of Informatics, University of Warsaw}}
\date{February 18, 2022}

\maketitle

\begin{abstract}
We prove the conjecture of Grosse-Kunstleve et al. that coordination sequences of periodic structures in
$n$-dimensional Euclidean space are rational. This has been recently proven by Nakamura et al.;
however, our proof is a straightforward application of classic techniques from automata theory.
\end{abstract}


\section{Introduction}
A periodic graph is a graph $(V,E)$ with a free action of $\bbZ^d$ on $V$, such that $V/\bbZ^d$ 
is finite, and such that for every $\{u,w\}\in E$ and $g \in \bbZ^d$, we have $\{gu,gw\} \in E$.
A coordination sequence of such a graph and vertex $v_0 \in V$ is $(c_0, c_1, c_2, \ldots)$ such
that $c_i$ is the number of vertices $v \in V$ such that the length of the shortest path from
$v_0$ to $v$ equals $i$.

From the paper by Chaim Goodman-Strauss and Neil Sloane \citeyear{colorbook} we have learned about the conjecture of \citet{conjcoord} that the
generating functions of coordination sequences are rational, which was open even in dimension 2. The generating function of a sequence 
$(c_0, c_1, c_2, \ldots)$ is defined as $q(z) = \sum_{z \geq 0} c_i z^i$; see e.g. \cite{concretemath} for an introduction to generating functions.
The following conditions are equivalent for polynomially bounded sequences:
\begin{itemize}
\item the generating function $q(z)$ is rational, i.e., $q(z)=p_1(z)/p_2(z)$ where $p_1$ and $p_2$ are polynomials,
\item the sequence $c_k$ can be computed using linear recursion: there exist $b_1, \ldots, b_t \in \bbR$ and $k_0$ such that for $k>k_0$ 
we have $c_k = \sum_{i=1}^t b_i c_{k-i}$,
\item the sequence $c_k$ is quasi-polynomial, i.e., we have $c_k = \sum a_i(k) k^i$, where $a_i(k)$ is a periodic function.
\end{itemize}

In this paper, we prove that this conjecture is a rather simple consequence of known techniques in automata theory.
The paper by \citet{colorbook} studies coordination sequences in specific two-dimensional cases, and 
contains a beautiful bibliography on coordination sequences.
Coordination sequences have been found in specific cases by \citet{okeeffe}, \citet{eon2002algebraic}, \citet{coordplane}, \citet{coordplane2}.
This conjecture has been recently proven by \citet{coordany}, however, we believe that our simple proof is still insightful,
as it shows another strong connection between apparently unrelated fields of crystallography and automata theory.
In particular, we believe that recent research in the computational aspects of Parikh images \cite{parikhlics,taming} may prove useful in crystallography.
Previously, strong connections were known between automata theory and periodic structures in hyperbolic geometry \cite{wpigroups}.
In Euclidean geometries, translations are commutative, so (in the 2D case) any translation by a vector in $\bbZ^2$ can be written as $a^n b^m$ where $a$ and $b$
are two translations corresponding to the periods. In the case of hyperbolic geometry, we also can use a finite number of generators, but this time, the order matters, so the structure of
a similar representation (word) is more complex. Formal language theory, and specifically the theory of finite automata, provides a framework to deal with such representations.
The tools for the Euclidean case are also provided by automata theory:
the Parikh image of a set of words is a subset of $\bbZ^d$ which 
contains the information about the possible numbers of occurences of every generator, abstracting from their order.

Our proof is based on the fact that, given a periodic graph, we can create a finite automaton
$A$ such that the Parikh image of $A$ is the set of all $(x_1,\ldots,x_n,z)$ such that there is a 
vertex $v$ of $V$ at coordinates $(x_1,\ldots,x_n)$ and that the length of the shortest path from $v_0$
to $v$ is at most $z$. By the Parikh theorem \cite{parikh}, this set is semilinear. It is known
that every semilinear set has an unambiguous representation as a union of disjoint unambiguous linear sets 
\cite{unambeilen,unambito,taming}. If $S \subseteq \bbZ^d$ is an unambiguous linear set, the generating function
$\sum_{(x_1,\ldots,x_n,d)\in S} z^d$ is rational, therefore, the generating function of every semilinear set is
rational.

The rest of the paper provides the necessary definitions and details of the proof.

\section{Semilinear sets}

Let $\bbN = \{0, 1, 2, 3, \ldots\}$.
A {\bf linear set} is a set $A \subset \bbZ^d$ of form $L(b,a_1,\ldots, a_k) = \{b + a_1 n_1 + a_2 n_2 + \ldots + a_k n_k: n_1, \ldots, n_k \in \bbN\}$.
A linear set $L(b,a_1,\ldots, a_k)$ is {\bf unambiguous} iff every $x \in L(b,a_1,\ldots,a_k)$ can be written in the
form $b + a_1 n_1 + a_2 n_2 + \ldots + a_k n_k$ in exactly one way.

A {\bf semilinear set} $A \subseteq \bbZ^d$ is a union of linear sets.

{\bf Example.} Consider the set $A : \{(a,b) \in \bbN^2 : a=b=0 \vee (a,b\geq 2 \wedge 2|a+b\}$. This set is a semilinear set, since it is a union of
$\{(0,0)\} = L((0,0))$ and $A_2 = L((2,2), (2,0), (1,1), (0,2))$. The set $A_2$ is not unambiguous, since e.g. $(4,4)$ can be represented in two ways, as
$(2,2)+(2,0)+(0,2)$ or $(2,2)+(1,1)+(1,1)$. However, $A_2$ can be written as a union of two unambiguous linear sets, for example, $A_2 = L((2,2), (2,0), (0,2)) \cup
L((3,3), (2,0), (0,2))$. Therefore, $A$ is a disjoint union of three unambiguous linear sets.

See \cite{hopcroft79,kozen97} for an introduction to automata theory.
A (non-deterministic) {\bf finite automaton} (NFA) over $\bbZ^d$ is a tuple $A = \{Q, I, F, \delta\}$, where $Q$ is a finite set of states,
$I \subseteq Q$ is the set of initial states, $F \subseteq Q$ is the set of final states, and $\delta \subseteq Q \times \bbZ^d \times Q$
is the (finite) set of transitions. For $(s,a,t) = d \in \delta$, we call $s=s(d)$ the source state of $d$, $t=t(d)$ the target state
of $d$, and $a=a(d)$ the output of $d$.

A {\bf run} of the finite automaton $A$ of length $k>0$ is a sequence $r = (r_1, \ldots, r_k) \in \delta^k$ (for some $k$) such that $s(r_1) \in I$, $s(r_{i+1}) = t(r_i)$
for $i=1, \ldots, k-1$, and $t(r_k) \in F$. An empty run (of length $k=0$) is also allowed if $I \cap F \neq \emptyset$
(intuitively, such runs start in an initial state and end in the same state without taking any transitions).
The Parikh image of a run $r$, denoted $\Psi(r)$, is $\sum_{i=1}^k a(r_i)$. The Parikh image of an automaton
$A$, denoted $\Psi(A)$, is the set of Parikh images of all its runs. We use the following fact:

\begin{theorem}[Parikh's Theorem, \cite{parikh}]
The Parikh image of a NFA is a semilinear set.
\end{theorem}

Note: in this paper $\delta \subseteq Q \times \bbZ^d \times Q$. The usual convention in automata theory is to take $\delta \subseteq Q \times \Sigma \times Q$, where 
$\Sigma$ is a finite alphabet $\Sigma = \{a_1, \ldots, a_d\}$, and the Parikh
image of a run is defined as a vector $v \in \bbN^d$ such that $v_i$ is the number of transitions labeled with $a_i$.
Our convention generalizes the usual convention -- just take $a_i$ to be the $i$-th unit vector. However, it is straightforward to generalize 
Parikh's Theorem to our setting \cite{parikhlics}. Indeed, our finite automaton can be seen as an automaton over the finite alphabet $\Sigma = \{a_1, \ldots, a_m\}$ where
$a_i$ are the vectors appearing as outputs. By the usual Parikh's theorem, the Parikh's image of $A$, $\Psi_1(M)$, is a semilinear subset of $\bbN^m$. Let $\phi: \bbN^m \ra \bbZ^k$ be
given as $\phi(n_1, \ldots, n_m) = \sum_{i=1,\ldots,m} a_i n_i$. The Parikh image in our sense, $\Psi(M)$, can then be written as $\phi(\Psi(m))$. Since $\phi$ is linear,
it maps semilinear sets to semilinear sets.

Parikh images and semilinear sets are a widely studied topic in automata theory; see \citet{taming} for a recent reference. The following characterization of semilinear sets
will be useful:

\begin{theorem}[\cite{unambeilen,unambito,taming}] \label{unamb}
Every semilinear set can be represented as a union of disjoint unambiguous linear sets.
\end{theorem}

Let $\pi_i: \bbZ^d \rightarrow \bbZ$ be the projection on the $i$-th coordinate. For a set $S \subseteq \bbZ^d$ such
that $\pi_i(s) \geq 0$ for all $s \in S$, let $q_i(S)(z) = \sum_{s\in S} z^{\pi_i(s)}$, that is,
the generating function of the sequence $(a_n)$ where $a_n$ is the number of elements of $S$ such that $\pi_i(s) = a_n$.

\begin{theorem}\label{semirat}
If $S$ is a semilinear set such that $\pi_i(s) \geq 0$ for all $s \in S$, the generating function $q_i(S)$ is rational.
\end{theorem}

\begin{proof}
If $S = L(b, a_1, \ldots, a_k)$ is an unambiguous linear set, every element of $S$ can be written as $b+\sum_{j=1..k} n_j a_j$ in exactly one way, so we have 
\begin{eqnarray*}
q_i(S) &=& \sum_{n_1, \ldots, n_k \in \bbN} z^{\pi_i \left(b+\sum_j n_j a_j \right)} = \sum_{n_1, \ldots, n_k \in \bbN} z^{\pi_i(b)} \prod_{j=1..k} z^{\pi_i(n_j a_j)} = \\
&=& z^{\pi_i(b)} \prod_{j=1}^k \sum_{n_j\in\bbN} z^{n_j \pi_i(a_j)} = z^{\pi_i(b)} / \prod_{j=1}^k (1-z^{\pi_i(a_j)}).
\end{eqnarray*}

In general, by Theorem \ref{unamb}, $S$ has an unambiguous representation $S = S_1 \cup \ldots \cup S_k$ where $S_i$ are disjoint
unambiguous linear sets, and thus, $q_i(S) = \sum_j q_i(S_j)$.
\end{proof}

\section{Periodic graphs}
Given the definition of a periodic graph presented in the introduction, let 
$\{v_1, \ldots, v_m\} = V_0 \subseteq V$ be a set containing exactly one element from every orbit of the action of $\bbZ^d$.
We will prove the following:

\begin{theorem}\label{regtoaut}
For $v_0 \in V_0$ and $i \in \{1, \ldots, m\}$, the set $C(v_0,i)$ of $\{(x_1, \ldots, x_d, y) \in \bbZ^d\times \bbN\}$
such that there exists a path from $v_0$ to $(x_1, \ldots, x_d) v_i \in G$ of length at most $y$
is a semilinear set.
\end{theorem}

\begin{proof}
We will construct a NFA $A$ such that $\Psi(A) = C(v_0, i)$.

Our set of states will be $Q=V_0$.

Our set of initial states will be $I=\{v_0\}$.

Our set of final states will be $F=\{v_i\}$.

Every orbit of edges contains exactly one edge from some $v_s \in V_0$ to $v'_t \in V$, which can be written as
$(x_1, \ldots, x_d) v_t$ for some $v_t \in V_0$. For every orbit of edges we include a transition
$d$ such that $s(d) = v_s$, $t(d) = v_t$, and $a(d) = \{x_1, \ldots, x_d, 1\}$.

For every $v_s \in V_0$ we also include a transition $d$ such that $s(d) = t(d) = v_s$ and $a(d) = \{0, \ldots, 0, 1\}$.

It is straightforward to show that $(x_1, \ldots, x_d, y) \in \Psi(A)$ iff there is a path from
$v_0$ to $T_1^{x_i} T_2^{x_2} \ldots T_d^{x_d} v_i \in G$ of length at most $y$, which proves the theorem.
\end{proof}

\begin{figure}
\centering
\includegraphics[width=\linewidth]{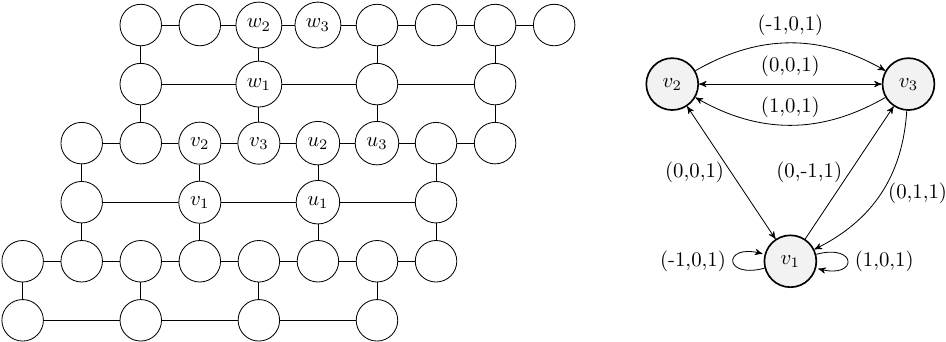}
\caption{{\bf Figure \ref{exfig}.} A periodic graph, and its automaton \label{exfig}}
\end{figure}

See Figure \ref{exfig} for an example periodic graph with $m=3$ orbits, and a diagram of 
the NFA from the proof of 
Theorem \ref{regtoaut}. 
We have $V_0 = \{v_1, v_2, v_3\}$. 
Translations work as follows: $u_i = (1,0)v_i$, $w_i = (0,1)v_i$.
In the diagram in Figure \ref{exfig}, a transition from $v_i$ to $v_j$ with output $a$ is depicted
as an arrow from $v_i$ to $v_j$ labelled $a$.

\newtheorem{corol}[theorem]{Corollary}

\begin{corol}\label{corlab}
The generating function of the coordination sequence of a periodic graph $(V,E)$ starting from $v_0$ is rational.
\end{corol}

\begin{proof}
By Theorem \ref{semirat}, the generating function $q_{d+1}(C(v_0,i))$ is rational. To count only paths of maximal
length exactly $y$, we need to divide this generating function by $1/(1-z)$. Therefore, the
generating function of the coordination sequence of $V$ from $v_0$ is rational. 
\end{proof}

{\bf Acknowledgments.}
This work has been supported by the National Science Centre, Poland, grant UMO-2019//35/B/ST6/04456.
We are grateful to the anonymous referees for their suggestions which improved this paper.

\end{document}